\documentclass[a4paper,UKenglish]{lipics-v2018}
 
\hideLIPIcs

\usepackage{microtype}
\theoremstyle{plain}
\newtheorem{proposition}[theorem]{Proposition}

\newcommand{\hide}[1]{}

\usepackage{scalerel}
\usepackage{relsize}
\usepackage{graphicx}
\usepackage{wrapfig}

\usepackage{tikz}
\usepackage{tikz-qtree}
\usetikzlibrary{positioning,shapes,snakes,arrows}

\usepackage{stmaryrd}
\newcommand{\sem}[1]{\llbracket #1\rrbracket}
\newcommand{\RR}{\mathbb R}
\newcommand{\II}{I}
\newcommand{\dd}{\mathrm d}
\newcommand{\dbeta}[2]{\beta_{#1,#2}}

\DeclareMathOperator*{\mulchsym}{\scalerel*{?}{\sum}}

\newcommand{\mulch}[2]{\left(\mulchsym\,\begin{matrix} {#1} \\ {#2} \end{matrix}\right)}
\newcommand{\mulchm}[1]{\left(\mulchsym\,\begin{matrix} #1 \end{matrix}\right)}

\newdimen\proofrulebreadth \proofrulebreadth=.05em
\newdimen\proofdotseparation \proofdotseparation=1.25ex
\newdimen\proofrulebaseline \proofrulebaseline=2ex
\newcount\proofdotnumber \proofdotnumber=3
\let\then\relax
\def\hfi{\hskip0pt plus.0001fil}
\mathchardef\squigto="3A3B
\newif\ifinsideprooftree\insideprooftreefalse
\newif\ifonleftofproofrule\onleftofproofrulefalse
\newif\ifproofdots\proofdotsfalse
\newif\ifdoubleproof\doubleprooffalse
\let\wereinproofbit\relax
\newdimen\shortenproofleft
\newdimen\shortenproofright
\newdimen\proofbelowshift
\newbox\proofabove
\newbox\proofbelow
\newbox\proofrulename
\def\shiftproofbelow{\let\next\relax\afterassignment\setshiftproofbelow\dimen0 }
\def\shiftproofbelowneg{\def\next{\multiply\dimen0 by-1 }%
\afterassignment\setshiftproofbelow\dimen0 }
\def\setshiftproofbelow{\next\proofbelowshift=\dimen0 }
\def\setproofrulebreadth{\proofrulebreadth}

\def\prooftree{
\ifnum  \lastpenalty=1
\then   \unpenalty
\else   \onleftofproofrulefalse
\fi
\ifonleftofproofrule
\else   \ifinsideprooftree
        \then   \hskip.5em plus1fil
        \fi
\fi
\bgroup
\setbox\proofbelow=\hbox{}\setbox\proofrulename=\hbox{}%
\let\justifies\proofover\let\leadsto\proofoverdots\let\Justifies\proofoverdbl
\let\using\proofusing\let\[\prooftree
\ifinsideprooftree\let\]\endprooftree\fi
\proofdotsfalse\doubleprooffalse
\let\thickness\setproofrulebreadth
\let\shiftright\shiftproofbelow \let\shift\shiftproofbelow
\let\shiftleft\shiftproofbelowneg
\let\ifwasinsideprooftree\ifinsideprooftree
\insideprooftreetrue
\setbox\proofabove=\hbox\bgroup$\displaystyle 
\let\wereinproofbit\prooftree
\shortenproofleft=0pt \shortenproofright=0pt \proofbelowshift=0pt
\onleftofproofruletrue\penalty1
}

\def\eproofbit{
\ifx    \wereinproofbit\prooftree
\then   \ifcase \lastpenalty
        \then   \shortenproofright=0pt  
        \or     \unpenalty\hfil         
        \or     \unpenalty\unskip       
        \else   \shortenproofright=0pt  
        \fi
\fi
\global\dimen0=\shortenproofleft
\global\dimen1=\shortenproofright
\global\dimen2=\proofrulebreadth
\global\dimen3=\proofbelowshift
\global\dimen4=\proofdotseparation
\global\count255=\proofdotnumber
$\egroup  
\shortenproofleft=\dimen0
\shortenproofright=\dimen1
\proofrulebreadth=\dimen2
\proofbelowshift=\dimen3
\proofdotseparation=\dimen4
\proofdotnumber=\count255
}

\def\proofover{
\eproofbit 
\setbox\proofbelow=\hbox\bgroup 
\let\wereinproofbit\proofover
$\displaystyle
}%
\def\proofoverdbl{
\eproofbit 
\doubleprooftrue
\setbox\proofbelow=\hbox\bgroup 
\let\wereinproofbit\proofoverdbl
$\displaystyle
}%
\def\proofoverdots{
\eproofbit 
\proofdotstrue
\setbox\proofbelow=\hbox\bgroup 
\let\wereinproofbit\proofoverdots
$\displaystyle
}%
\def\proofusing{
\eproofbit 
\setbox\proofrulename=\hbox\bgroup 
\let\wereinproofbit\proofusing
\kern0.3em$
}

\def\endprooftree{
\eproofbit 
  \dimen5 =0pt
\dimen0=\wd\proofabove \advance\dimen0-\shortenproofleft
\advance\dimen0-\shortenproofright
\dimen1=.5\dimen0 \advance\dimen1-.5\wd\proofbelow
\dimen4=\dimen1
\advance\dimen1\proofbelowshift \advance\dimen4-\proofbelowshift
\ifdim  \dimen1<0pt
\then   \advance\shortenproofleft\dimen1
        \advance\dimen0-\dimen1
        \dimen1=0pt
        \ifdim  \shortenproofleft<0pt
        \then   \setbox\proofabove=\hbox{%
                        \kern-\shortenproofleft\unhbox\proofabove}%
                \shortenproofleft=0pt
        \fi
\fi
\ifdim  \dimen4<0pt
\then   \advance\shortenproofright\dimen4
        \advance\dimen0-\dimen4
        \dimen4=0pt
\fi
\ifdim  \shortenproofright<\wd\proofrulename
\then   \shortenproofright=\wd\proofrulename
\fi
\dimen2=\shortenproofleft \advance\dimen2 by\dimen1
\dimen3=\shortenproofright\advance\dimen3 by\dimen4
\ifproofdots
\then
        \dimen6=\shortenproofleft \advance\dimen6 .5\dimen0
        \setbox1=\vbox to\proofdotseparation{\vss\hbox{$\cdot$}\vss}%
        \setbox0=\hbox{%
                \advance\dimen6-.5\wd1
                \kern\dimen6
                $\vcenter to\proofdotnumber\proofdotseparation
                        {\leaders\box1\vfill}$%
                \unhbox\proofrulename}%
\else   \dimen6=\fontdimen22\the\textfont2 
        \dimen7=\dimen6
        \advance\dimen6by.5\proofrulebreadth
        \advance\dimen7by-.5\proofrulebreadth
        \setbox0=\hbox{%
                \kern\shortenproofleft
                \ifdoubleproof
                \then   \hbox to\dimen0{%
                        $\mathsurround0pt\mathord=\mkern-6mu%
                        \cleaders\hbox{$\mkern-2mu=\mkern-2mu$}\hfill
                        \mkern-6mu\mathord=$}%
                \else   \vrule height\dimen6 depth-\dimen7 width\dimen0
                \fi
                \unhbox\proofrulename}%
        \ht0=\dimen6 \dp0=-\dimen7
\fi
\let\doll\relax
\ifwasinsideprooftree
\then   \let\VBOX\vbox
\else   \ifmmode\else$\let\doll=$\fi
        \let\VBOX\vcenter
\fi
\VBOX   {\baselineskip\proofrulebaseline \lineskip.2ex
        \expandafter\lineskiplimit\ifproofdots0ex\else-0.6ex\fi
        \hbox   spread\dimen5   {\hfi\unhbox\proofabove\hfi}%
        \hbox{\box0}%
        \hbox   {\kern\dimen2 \box\proofbelow}}\doll%
\global\dimen2=\dimen2
\global\dimen3=\dimen3
\egroup 
\ifonleftofproofrule
\then   \shortenproofleft=\dimen2
\fi
\shortenproofright=\dimen3
\onleftofproofrulefalse
\ifinsideprooftree
\then   \hskip.5em plus 1fil \penalty2
\fi
}

\newcommand*{\mlstinline}[1]{\text{\lstinline|#1|}}

\title{The Beta-Bernoulli process and algebraic effects}

\author{Sam Staton\ $^1$,\; Dario Stein\ $^1$,\; Hongseok Yang\ $^2$,\; Nathanael L.\ Ackerman\ $^3$,\; Cameron E.\ Freer\ $^4$,\; and Daniel M.\ Roy\ $^5$}{1 $\ $ Univ.~Oxford, $\ \,$ 2$\ $ KAIST,$\ \,$ 3$\ $ Harvard Univ.,$\ \,$ 4$\ $ Borelian, $\ \,$ 5$\ $ Univ.~Toronto}{}{}{}
	
\authorrunning{S Staton, D Stein, H Yang, NL Ackerman, CE Freer, DM Roy} 

\Copyright{Sam Staton, Dario Stein, Hongseok Yang, Nathanael L.\ Ackerman, Cameron E.\ Freer, Daniel M.\ Roy}

\subjclass{Theory of computation → Probabilistic computation}

\relatedversion{This paper is to appear in Proc.~ICALP 2018~\cite{Staton-BB18-short}. It is here extended with a short appendix.}
\EventEditors{Ioannis Chatzigiannakis, Christos Kaklamanis, D\'{a}niel Marx, and Don Sannella}
\EventNoEds{4}
\EventLongTitle{45th International Colloquium on Automata, Languages, and Programming (ICALP 2018)}
\EventShortTitle{ICALP 2018}
\EventAcronym{ICALP}
\EventYear{2018}
\EventDate{July 9--13, 2018}
\EventLocation{Prague, Czech Republic}
\EventLogo{eatcs}
\SeriesVolume{107}
\ArticleNo{344}

\nolinenumbers

\newcommand*{\defeq}{\stackrel{\text{def}}=}

\begin{document}

\maketitle

\begin{abstract}
In this paper we use the framework of algebraic effects from programming language theory to analyze the Beta-Bernoulli process, a standard building block in Bayesian models. 
Our analysis reveals the importance of abstract data types, and two types of program equations, called commutativity and discardability. 
We develop an equational theory of terms that use the Beta-Bernoulli process, 
and show that the theory is complete with respect to the measure-theoretic semantics, and also in the syntactic sense of Post. 
Our analysis has a potential for being generalized to other stochastic processes relevant to Bayesian modelling, yielding new understanding of these processes from the perspective of programming.
 \end{abstract}

\section{Introduction}
\label{sec:intro}

From the perspective of programming, a family of Boolean random processes is implemented by a module that supports the following interface:
\begin{lstlisting}
module type ProcessFactory = sig type process
                                  val new : $H$ -> process
                                  val get : process -> bool  end
\end{lstlisting}
where $H$ is some type of hyperparameters. 
Thus one can initialize a new process, 
and then get a sequence of Booleans from that process. 
The type of processes is kept abstract so that any internal state or representation is hidden. 

One can analyze a module extensionally in terms of the properties of its interactions with a client program. 
In this paper, we perform this analysis for the Beta-Bernoulli process, an important building block in Bayesian models. 
We completely axiomatize its equational properties, using the formal framework of algebraic effects~\cite{pp-algop-geneff}. 

The following modules are our leading examples. (Here \lstinline|flip(r)| tosses a coin with bias \lstinline|r|.)
\begin{figure}[!htb]
\begin{minipage}{.6\textwidth}
\begin{lstlisting}
module Polya = (struct
  type process = (int * int) ref 
  let new(i,j) = ref (i,j)
  let get p = let (i,j) = !p in
    if flip(i/(i+j)) then p := (i+1,j); true
    else p := (i,j+1); false  end : ProcessFactory)
\end{lstlisting}
\end{minipage}
\,
\begin{minipage}{.4\textwidth}
\begin{lstlisting}
module BetaBern = (struct
  type process = real
  let new(i,j) = sample_beta(i,j)
  let get(r) = flip(r)

end : ProcessFactory)
\end{lstlisting}
\end{minipage}
\end{figure}

The left-hand module, \lstinline|Polya|, is an implementation of P\'olya's urn. 
An urn in this sense is a hidden state which contains $i$-many balls marked \lstinline|true| and $j$-many balls marked \lstinline|false|. 
To sample, we draw a ball from the urn at random; before we tell what we drew, we put back the ball we drew as well as an identical copy of it. The contents of the urn changes over time.

\begin{wrapfigure}[7]{r}{3.5cm}\vspace{-4mm}
\mbox{\hspace{-5mm}%
\input{beta-bern}}%
\end{wrapfigure} 
 The right-hand module, \lstinline|BetaBern|, is based on the beta distribution. This is the probability measure on the unit interval $[0, 1]$ that measures the bias of a random source (such as a potentially unfair coin) from which \lstinline|true| has been observed $(i-1)$ times and \lstinline|false| has been observed $(j-1)$ times, as illustrated on the right. For instance $\mathsf{beta(2,2)}$ describes the situation where we only know that neither \lstinline|true| nor \lstinline|false| are impossible; while in $\mathsf{beta(3,2)}$ we are still ignorant but we believe that \lstinline|true| is more likely. 

It turns out that these two modules have the same observable behaviour. This essentially follows from de Finetti's theorem~(e.g.~\cite{schervish}), but rephrased in programming terms.
The equivalence makes essential use of type abstraction: if we could look into the urn, or ask precise questions about the real number, the modules would be distinguishable. 

The module \lstinline|Polya| has a straightforward operational semantics (although we don't formalize that here). By contrast, \lstinline|BetaBern| has a straightforward
denotational semantics~\cite{Kozen81}. In Section~\ref{sec:presentation-beta-bernoulli}, we provide an axiomatization of equality, which is sound by both accounts. We show completeness of our 
axiomatization with respect to the denotational semantics of \lstinline|BetaBern| (\S\ref{sec:completeness},~Thm.~\ref{thm:model-completeness}). We use this to show that the axiomatization is in fact 
syntactically complete (\S\ref{sec:extensionality-syntactic-completeness},~Cor.~\ref{cor:hpcomplete-bb}), which means it is complete with respect to \emph{any} semantics. 

\medskip 
For the remainder of this section, we give a general introduction to our axioms. 
\medskip

\noindent{\bf Commutativity and discardability.}\
Commutativity and discardability are important program equations~\cite{fuhrmann} that are closely related, we argue, to exchangeability in statistics.
\begin{itemize}
\item \emph{Commutativity} is the requirement that when \lstinline|x| is not free in \lstinline|u| and \lstinline|y| is not free in \lstinline|t|,
\[
\Big(\mlstinline{let x = t in let y = u in v}\Big) \quad=\quad
\Big(\mlstinline{let y = u in let x = t in v}\Big).
\]
\item \emph{Discardability} is the requirement that when \lstinline|x| is not free in \lstinline|u|,
$\Big(\mlstinline{let x = t in u}\Big) =
\Big(\mlstinline{u}\Big)$.
\end{itemize}

Together, these properties say that data flow, rather than the control flow, is what matters. For example, in a standard programming language, the purely functional total expressions are commutative and discardable. By contrast, expressions that write to memory are typically not commutative or discardable
(a simple example is \lstinline|t|$ = $\lstinline|u|$ = $\lstinline|a++|, \lstinline|v|$ = $\lstinline|(x,y)|).  A simple example of a commutative and discardable operation is a coin toss: we can reorder the outcomes of tossing a single coin, 
and we can drop some of the results (unconditionally) without changing the overall statistics.

We contend that commutativity and discardability of program expressions is very close to the basic notion of exchangeability of infinite sequences, which is central to Bayesian statistics. 
Informally, an infinite random process, such as an infinite random sequence, is said to be exchangeable if one can reorder and discard draws without changing the overall statistics. 
(For more details on exchangeable random processes in probabilistic programming languages, see \cite{XRP-PPS2016,XRPDA-PPS2017}, and the references therein.)
A client program for the \lstinline|BetaBern| module is clearly exchangeable in this sense: this is roughly Fubini's theorem.
For the \lstinline|Polya| module, an elementary calculation is needed: it is not trivial because memory is involved.

\medskip
\noindent{\bf Conjugacy.}\
Besides exchangeability, the following conjugacy equation is crucial:
\begin{align*}
& \Big(\mlstinline{let p=M.new(i,j) in (M.get(p), p)}\Big) 
\\
& \quad {} = \Big(\mlstinline{if flip(i/(i+j)) then (true, M.new(i+1,j)) else (false, M.new(i,j+1))}\Big).
\end{align*}
This is essentially the operational semantics of the \lstinline|Polya| module, and 
from the perspective of \lstinline|BetaBern| it is the well-known conjugate-prior relationship between the Beta and Bernoulli distributions. 

\medskip
\noindent{\bf Finite probability.}\
In addition to exchangeability and conjugacy, we include the standard equations of finite, discrete, rational probability theory. 
To introduce these, suppose that we have a module 
\begin{lstlisting}
Bernoulli : sig val get : int * int -> bool end
\end{lstlisting}
which is built so that \lstinline|Bernoulli.get($i$,$j$)| samples 
\emph{with single replacement} from an urn with $i$-many balls marked \lstinline|true| and $j$-many balls marked \lstinline|false|.
(In contrast to P\'olya's urn, the urn in this simple scheme does not change over time.)
So \lstinline|Bernoulli.get($i$,$j$)|$\ =\ $\lstinline|flip($\frac i {i+j}$)|.
This satisfies certain laws, first noticed long ago by Stone~\cite{stone}, and 
recalled in \S\ref{sec:algebraic-Bernoulli}. 

In summary, our main contribution is that these axioms --- exchangeability, conjugacy, and finite probability --- 
entirely determine the equational theory of the Beta-Bernoulli process, in the following sense:
\begin{itemize}
\item \emph{Model completeness:} Every equation that holds in the measure theoretic interpretation is derivable from our axioms (Thm.~\ref{thm:model-completeness});
\item \emph{Syntactical completeness:} Every equation that is not derivable from our axioms is inconsistent with finite discrete probability (Cor.~\ref{cor:hpcomplete-bb}). 
\end{itemize}
We argue that these results open up a new method for analyzing Bayesian models, based on algebraic effects (see \S\ref{sec:conc} and \cite{XRPDA-PPS2017}\footnote{This paper formalizes and proves a conjecture from~\cite{XRPDA-PPS2017}, which is an unpublished abstract.}). 


\section{An algebraic presentation of the Beta-Bernoulli process}
\label{sec:presentation-beta-bernoulli}

In this section, we present syntactic rules for
well-formed client programs of the Beta-Bernoulli module, and
axioms for deriving equations on those programs.

\subsection{An algebraic presentation of finite probability}
\label{sec:algebraic-Bernoulli}
\newcommand{\rch}[2]{\mathop{{}_{#1}\!{?}\!_{#2}}}
\newcommand{\pch}[1]{\mathop{?_{#1}}}
\newcommand{\tj}[2]{#1\vdash #2}
\newcommand{\pj}[3]{#1\mathop|#2\vdash #3}
\newcommand{\NN}{\mathbb{N}}
Recall the module \lstinline|Bernoulli| from the introduction
which provides a method of sampling with odds $(i:j)$. 
We will axiomatize its equational properties. Algebraic effects provide a way to 
axiomatize the specific features of this module while putting aside the general properties of programming languages, such as $\beta/\eta$ laws. 
In this situation the basic idea is that each module induces a binary operation 
$\rch ij$ on programs by 
\[t\rch ij u\ \defeq \ 
\mlstinline{if Bernoulli.get($i$,$j$) then $\ t\ $ else $\ u$}\text.\]
Conversely, given a family of binary operations $\rch ij$, we can recover $\mlstinline{Bernoulli.get($i$,$j$)} = \mlstinline{true}\rch i j \mlstinline{false}$.
So to give an equational presentation of the \lstinline|Bernoulli| module we 
give a equational presentation of the binary operations $\rch ij$. 
A full programming language will have other constructs and $\beta\eta$-laws but it is routine to combine these with an algebraic theory of effects~(e.g.~\cite{ahmanstaton,jsv,kammarplotkin,pretnar}).
\begin{definition}
\label{def:rational-convex}
The \emph{theory of rational convexity} is the first-order algebraic theory with binary operations $\rch ij$ for all $i,j\in \NN$ such that $i+j>0$, subject to the axiom schemes
\begin{align*}
\tj{w,x,y,z}{&(w \rch i j x) \rch {i+j}{k+l}(y \rch k l z)
=
(w \rch i k y) \rch {i+k}{j+l}(x \rch j l z) }
\\
\tj{x,y}{&x \rch i j y = y \rch j i x}
\qquad \qquad
\tj{x,y}{x \rch {i} {0} y = x}
\qquad\qquad\tj{x}{x \rch {i} {j} x = x}
\end{align*}
\end{definition}

Commutativity 
${(w\rch ij x)\rch kl (y \rch ij z) =(w\rch k l y)\rch i j (x \rch kl z)}$ 
of operations $\rch kl$ and $\rch ij$ is a derivable equation, 
and so is scaling $x\rch {ki} {kj} y =x\rch {i} {j} y$ for $k>0$. Commutativity and discardability ($x\rch ij x=x$) in this algebraic sense (cf. \cite{linton-comm,modes}) precisely correspond to the program equations in Section~\ref{sec:intro} (see also~\cite{kammarplotkin}). The theory first appeared in~\cite{stone}.

\subsection{A parameterized algebraic signature for Beta-Bernoulli}
\label{sec:pat-beta-bern}
In the theory of convex sets, the parameters $i,j$ for \lstinline|get| 
range over the integers. These integers are not a first class concept in our equational 
presentation: we did not axiomatize integer arithmetic. 
However, in the Beta-Bernoulli process, or any module \lstinline|M| for the \lstinline|ProcessFactory| interface,
it is helpful to understand
the parameters to \lstinline|get| as abstract, and \lstinline|new|
as generating such parameters. To interpret this, we treat these parameters to \lstinline|get| as first class. There are still hyperparameters to \lstinline|new|, which we do not treat as first class here. (In a more complex hierarchical system with hyperpriors, we might treat them as first class.)

As before, to avoid studying an entire programming language, we look at the constructions
\begin{align*}
&\nu_{i,j}p.t\ \defeq \ \mlstinline{let $\ p$=M.new($i$,$j$) in  $\ t$}
&t\pch p u\ \defeq\ \mlstinline{if M.get($p$) then  $\ t \ $  else  $\ u$}
\end{align*}
There is nothing lost by doing this, because we can recover $\mlstinline{M.new($i$,$j$)} = \nu_{i,j}p.\,p$ and $\mlstinline{M.get($p$)} = \mlstinline{true}\pch p \mlstinline{false}$. 
In the terminology of~\cite{pp-algop-geneff}, these would be called the `generic effects' of the algebraic operations $\nu_{i,j}$ and $\pch p$. Note that $\pch p$ is a parameterized binary operation. Formally, our syntax now has two kinds of variables: $x$, $y$ as before, ranging over continuations, and now also $p$, $q$ ranging over parameters. We notate this by having contexts with two zones, and write $x:n$ if $x$ expects $n$ parameters.
\begin{definition}
\label{def:pat-beta-bern}
The term formation rules for the theory of Beta-Bernoulli are:
\begin{align*}
&\begin{prooftree}-\justifies
\pj{\Gamma}{\Delta,x:m,\Delta'}{x(p_1\dots p_m)}
\using{\mbox{\small{$(p_1\dots p_m\in \Gamma)$}}}
\end{prooftree}
&
\begin{prooftree}
\pj{\Gamma,p}{\Delta} t
\justifies
\pj{\Gamma}{\Delta} {\nu_{i,j} p.t}
\using {\mbox{\small{$(i,j>0)$}}}
\end{prooftree}
\\[10pt]
&\begin{prooftree}
\pj{\Gamma}\Delta t
\quad
\pj\Gamma\Delta u
\justifies
\pj \Gamma\Delta {t\pch{p}u}
\using {\mbox{\small $(p\in\Gamma)$}}
\end{prooftree}
&
\begin{prooftree}
\pj{\Gamma}\Delta t
\quad
\pj\Gamma\Delta u
\justifies
\pj \Gamma\Delta {t\rch ij u}
\using {\mbox{\small{$(i+j>0)$}}}
\end{prooftree}
\end{align*}
where $\Gamma$ is a parameter context of the form $\Gamma=(p_1,\ldots, p_\ell)$ 
and $\Delta$ is a context of the form $\Delta=(x_1\colon m_1,\ldots, x_k\colon m_k)$. Where ${x\colon 0}$, we often write $x$ for $x()$. For the sake of a well-defined notion of dimension in \ref{sec:completenessproof}, we disallow the formation of $\nu_{i,0}$ and $\nu_{0,i}$. 
\end{definition}

We work up-to $\alpha$-conversion and substitution of terms for variables must avoid unintended capture of free parameters. For example, substituting $x \pch p y$ for $w$ in $\nu_{1,1}p.w$ yields $\nu_{1,1}q.(x\pch p y)$,
while substituting $x \pch p y$ for $z(p)$ in $\nu_{1,1}p.z(p)$ yields $\nu_{1,1}p.(x\pch p y)$.

\subsection{Axioms for Beta-Bernoulli}
\label{sec:axioms-beta-bernoulli}
The axioms for the Beta-Bernoulli theory comprise the axioms for rational convexity
(Def.~\ref{def:rational-convex}) together with the following axiom schemes.

\begin{description}
\item[Commutativity.] All the operations commute with each other:
\begin{align}
\pj{p,q}{w,x,y,z:0}{\;&(w\pch q x)\pch p(y\pch q z)=(w\pch p y)\pch q(x\pch p z)}
	\label{eqn:commpchpch} \tag{C1}
\\
\nonumber
\pj{-}{x\colon 2}{\;&\nu_{i,j}p.(\nu_{k,l}q.x(p,q))=\nu_{k,l}q.(\nu_{i,j}p.x(p,q))}
	\label{eqn:commnunu} \tag{C2}
\\
\nonumber
\pj{q}{x,y\colon 1}{\;&\nu_{i,j}p.(x(p)\pch q y(p))=(\nu_{i,j}p.x(p))\pch q (\nu_{i,j}p.y(p))}
	\label{eqn:commnupch} \tag{C3}
\\
\nonumber
\pj{-}{x,y:1}{\;&\nu_{i,j}p.(x(p)\rch kl y(p))=(\nu_{i,j}p.x(p))\rch kl (\nu_{i,j}p.y(p))}
	\label{eqn:commnurch} \tag{C4}
\\
\nonumber
\pj{p}{w,x,y,z:0}{\;&(w\rch ij x)\pch p(y\rch ij z)=(w\pch p y)\rch ij(x\pch p z)}
	\label{eqn:commpchrch} \tag{C5}
\end{align}
\item[Discardability.] All operations are idempotent:
\begin{align*}
	\pj{-}{x\colon 0}{(\nu_{i,j}p.x)=x}&&\quad
\pj{p}{x\colon 0}{x\pch p x=x}
\tag{D1--2}\end{align*}
\item[Conjugacy.]
\begin{align}
\pj{-}{x,y:1}{&\nu_{i,j}p.(x(p)\pch p y(p))=(\nu_{i+1,j}p.x(p))\rch ij (\nu_{i,j+1}p.y(p)})
	\label{eqn:conjugacy}\tag{Conj}
\end{align}
\end{description}
A theory of equality for terms in context is built, as usual, by closing the axioms 
under substitution, congruence, reflexivity, symmetry and transitivity. It immediately follows from conjugacy and discardability that $x \rch ij y $ is definable as $\nu_{i,j}p.(x\pch p y)$ for $i,j>0$. 

As an example, consider 
$t(r) = (r \pch p x) \pch p (y \pch p r)$ that represents tossing a coin with bias $p$ twice, continuing with $x$ or $y$ if the results are different, or with $r$ otherwise. One can show that $x \rch 1 1 y$ is a unique fixed point of $t$, i.e. $x \rch 1 1 y = t(x \rch 1 1 y)$; see \S\ref{app:derivation} 
for detail.
This is exactly von Neumann's trick \cite{vonneumann} to simulate a fair coin toss with a biased one. 

(For more details on the general axiomatic framework with parameters, see~\cite{s-pred-logic,s-instances}, where it is applied to predicate logic, $\pi$-calculus, and other effects.)


\section{A complete interpretation in measure theory}
\label{sec:completeness}
In this section we give an interpretation of terms using measures and integration operators, the standard formalism for probability theory (e.g.~\cite{pollard,schervish}), and we show that this interpretation is complete (Thm.~\ref{thm:model-completeness}). 
Even if the reader is not interested in measure theory, they may still find value in the syntactical results of \S\ref{sec:extensionality-syntactic-completeness} which we prove using this completeness result.

\medskip
By the Riesz–Markov–Kakutani representation theorem, there are two equivalent ways to view probabilistic programs: as probability kernels and as linear functionals. Both are useful.

\paragraph*{Programs as probability kernels.}
Forgetting about abstract types for a moment, terms in the \lstinline|BetaBern| module are first-order probabilistic programs. So we have a standard denotational semantics due to \cite{Kozen81} where terms are interpreted as probability kernels and $\nu$~as integration. Let $\II = [0,1]$ denote the unit interval. We write $\beta_{i,j}$ for the $\textrm{Beta}(i,j)$-distribution on $\II$, which is given by the density function $p\mapsto \frac{1}{B(i,j)}p^{i-1}(1-p)^{j-1}$, where $B(i,j)=\frac{(i-1)!(j-1)!}{(i+j-1)!}$ is a normalizing constant.

For contexts of the form $\Gamma = (p_1,\ldots, p_\ell)$ and $\Delta = (x_1 : m_1, \ldots, x_k : m_k)$, we let
$\sem{\Delta} \defeq \sum_{i=1}^k \II^{m_i}$ consist of a copy of $\II^{m_i}$ for every variable $x_i : m_i$.
This has a $\sigma$-algebra $\Sigma(\sem\Delta)$ generated by the Borel sets. We interpret terms $\pj \Gamma \Delta t$ as probability kernels $\sem{t} : \II^\ell\times\Sigma( \sem{\Delta})\to [0,1]$ inductively,
for $\vec p\in\II^\ell$ and $U\in \Sigma( \sem{\Delta})$ :
\begin{align*}
& \sem{x_i(p_{j_1},\ldots,p_{j_m})}(\vec p, U) = 1\text{ if $(i,p_{j_1}\ldots p_{j_m})\in U$, $0$ otherwise}
\\
& \sem{u \rch ij v}(\vec p,U) = \tfrac 1 {i+j}\Big(i(\sem{u}(\vec p,U)) + j(\sem{v}(\vec p,U))\Big) 
\\
& \sem{u \pch{p_j} v}(\vec p,U) = p_j(\sem{u}(\vec p,U)) + (1-p_j)(\sem{v}(\vec p,U))
\\
& \sem{\nu_{i,j}q.t}(\vec p,U) = \int_0^1 \sem{t}((\vec p,q),U) \, \beta_{i,j}(\mathrm dq)
\ \ \ \ \Big[=\int_0^1 \sem{t}((\vec p,q),U) \ \tfrac{1}{B(i,j)}q^{i-1}(1-q)^{j-1}\  \mathrm dq\Big]
\end{align*}
\begin{proposition}\label{prop:model-soundness1}
The interpretation is sound: if $\pj{\Gamma}{\Delta}t=u$ 
is derivable then $\sem t=\sem u$ as probability kernels $\sem \Gamma\times \Sigma(\sem\Delta)\to [0,1]$. 
\end{proposition}
\begin{proof}[Proof notes]
One must check that the axioms are sound under the interpretation. 
Each of the axioms are elementary facts about probability. 
For instance, commutativity \eqref{eqn:commnunu} amounts to Fubini's theorem, and
the conjugacy axiom \eqref{eqn:conjugacy} is the well-known conjugate-prior relationship of Beta- and Bernoulli distributions. 
\end{proof}

\paragraph*{Interpretation as functionals}
We write $\RR^{\II^m}$ for the vector space of continuous 
functions $\II^m \to \RR$, endowed with the supremum norm. 
Given a probability kernel $\kappa:\II^\ell \times \Sigma\big(\sum_{j=1}^k\II^{m_j}\big)\to [0,1]$ 
and $\vec p\in\II^\ell$, we define a linear map
$\phi_{\vec p}: \RR^{\II^{m_1}}\times \dots\times \RR^{\II^{m_k}}\to \RR$, by considering $\kappa$ as an integration operator:
\[\textstyle\phi_{\vec p}(f_1\ldots f_k)= \int f_j(r_1\ldots r_{m_j})\ \kappa(\vec p, \dd(j,r_1\ldots r_{m_j}))\]
Here $\phi_{\vec p}$ are unital ($\phi(\vec 1)=1$) and positive ($\vec f\geq 0\implies \phi(\vec f)\geq 0$).

When $\kappa=\sem t$, this $\phi_{\vec p}(\vec f)$ is moreover continuous 
in $\vec p$, and hence 
a unital positive linear map $\phi: \RR^{\II^{m_1}}\times \dots\times \RR^{\II^{m_k}}\to \RR^{\II^\ell}$ 
\cite[Thm.~5.1]{FurberJacobs}.
It is informative to spell out the interpretation of terms ${\pj {p_1,\ldots,p_\ell} {x_1 : m_1, \ldots, x_k : m_k} {t}}$ as maps
$ \sem{t} : \RR^{\II^{m_1}} \times \ldots \times \RR^{\II^{m_k}} \to \RR^{\II^\ell}$
since it fits the algebraic notation: we may think of the variables $x\colon m$ as ranging over functions~$\RR^{\II^m}$.
\begin{proposition}\label{prop:model-soundness2}
The functional interpretation is inductively given by 
\begin{align*}
&\sem{x_i(p_{j_1},\ldots,p_{j_m})}(\vec f)(\vec p) = f_i(p_{j_1},\ldots,p_{j_m}) \\
&\sem{u \rch ij v}(\vec f)(\vec p) = \tfrac 1 {i+j}\Big(i(\sem{u}(\vec f)(\vec p)) + j(\sem{v}(\vec f)(\vec p))\Big) \\
&\sem{u \pch{p_j} v}(\vec f)(\vec p) = p_j(\sem{u}(\vec f)(\vec p)) + (1-p_j)(\sem{v}(\vec f)(\vec p)) \\
&\sem{\nu_{i,j}q.t}(\vec f)(\vec p) = \int_0^1 \sem{t}(\vec f)(\vec p,q) \, \beta_{i,j}(\mathrm dq)
\end{align*}
\end{proposition}
For example, $\sem{\pj-{x,y\colon 0}{x\rch 11 y}}:\RR\times \RR\to \RR$ is the function 
$(x,y)\mapsto \frac 12(x+y)$, and 
$\sem{\pj-{x\colon 1}{\nu_{1,1}p.x(p)}}:\RR^\II\to\RR$ is the integration functional,
$f\mapsto \int_0^1 f(p)\,\dd p$.  

(We use the same brackets $\sem-$ for both the measure-theoretic and the functional interpretations; the intended 
semantics will be clear from context.)

\subsection{Technical background on Bernstein polynomials}
\begin{definition}[Bernstein polynomials]
For $i=0,\ldots,k$, we define the $i$-th basis Bernstein polynomial $b_{i,k}$ of degree $k$ as
$b_{i,k}(p) = \binom k i p^{k-i}(1-p)^i$.
For a multi-index $I = (i_1,\ldots,i_\ell)$ with $0 \leq i_j \leq k$, we let $b_{I,k}(\vec p) = b_{i_1,k}(p_1)\cdots b_{i_\ell,k}(p_\ell)$. A Bernstein polynomial is a linear combination of Bernstein basis polynomials.
\end{definition}
The family $\{ b_{i,k} : i = 0,\ldots,k \}$ is indeed a basis of the polynomials of maximum degree $k$ and also a partition of unity, i.e. $\sum_{i=0}^k b_{i,k} = 1$. Every Bernstein basis polynomial of degree $k$ can be expressed as a nonnegative rational linear combination of degree $k+1$ basis polynomials.

The density function of the distribution $\beta_{i,j}$ on $[0,1]$ for $i,j>0$ is proportional to a Bernstein basis polynomial of degree $i+j-2$. We can conclude that the measures $\{ \beta_{i,j} : i,j>0, i + j = n \}$ are linearly independent for every $n$. In higher dimensions, the polynomials $\{b_{I,k}\}$ are linearly independent for every $k$.
Moreover, products of beta distributions $\beta_{i_r,j_r}$ are linearly independent as long as $i_r + j_r = n$ holds for some common $n$. This will be a key idea for normalizing Beta-Bernoulli terms.

\subsection{Normal forms and completeness}
For the completeness proof of the measure-theoretic model, we proceed as follows: To decide $\pj{\Gamma}{\Delta}{ t = u}$ for two terms $t,u$, we transform them into a common normal form whose interpretations can be given explicitly. We then use a series of linear independence results to show that if the interpretations agree, the normal forms are already syntactically equal. \\
Normalization happens in three stages. 
\begin{itemize}
\item If we think of a term as a syntax tree of binary choices and $\nu$-binders, we use the conjugacy axiom to push all occurrences of $\nu$ towards the leaves of the tree.
\item We use commutativity and discardability to stratify the use of free parameters~$\pch p$.
\item The leaves of the tree will now consist of chains of $\nu$-binders, variables and ratio choices $\rch i j$. Those can be collected into a canonical form.
\end{itemize}
We will describe these normalization stages in reverse order because of their increasing complexity.

\subsubsection{Stone's normal forms for rational convex sets}
Normal forms for the theory of rational convex sets have been described by Stone~\cite{stone}. 
We note that if $\pj- {x_1\ldots x_k : 0}t$ is a term in the theory of rational convex sets (Def.~\ref{def:rational-convex}) then
$\sem t:\RR^k\to \RR$ is a unital positive linear map that takes rationals to rationals. 
From the perspective of measures, this corresponds to a categorical distribution with $k$ categories.
\begin{proposition}[Stone]
The interpretation exhibits a bijective correspondence between 
terms $\pj- {x_1\ldots x_k : 0}t$ built from $\rch ij$, modulo equations, and unital positive 
linear maps $\RR^k\to \RR$ that take rationals to rationals. 
\end{proposition}
For instance, the map $\phi(x,y,z)=\frac 1 {10}(2x+3y+5z)$ is unital positive linear, 
and arises from the term $t\defeq x\rch 28 (y\rch 35 z)$.
This is the only term that gives rise to the $\phi$, modulo equations. 
In brief, one can recover $t$ from $\phi$ by looking at $\phi(1,0,0)=\frac 2{10}$, then $\phi(0,1,0)=\frac 3{10}$, then $\phi(0,0,1)=\frac 5{10}$. We will write 
$\mulchm{x_1 & \ldots  & x_k \\
        w_1 & \ldots & w_k}$
for the term corresponding to the linear map $(x_1\ldots x_k)\mapsto \frac {1}{\sum_{i=1}^kw_k}(w_1x_1+\cdots +w_kx_k)$.
These are normal forms for the theory of rational convex sets.

\hide{We define a \emph{multichoice term} inductively. For a term $x$ and a positive weight $w \in \mathbb N$, let
\[ \mulch{x}{w} = x \]
For terms $x_1,\ldots,x_n$ and weights $w_1,\ldots,w_n$, not all zero, define
\[
\mulchm{x_1 & \ldots  & x_n \\
        w_1 & \ldots & w_n}
= x_1 \rch{w_1}{w_2 + \ldots + w_n} 
\mulchm{x_2 & \ldots  & x_n \\
        w_2 & \ldots & w_n}
\]
if $w_2 + \ldots + w_n \neq 0$, or $x_1$ otherwise. Stone showed that the columns of the multichoice can be permuted, weights for repeated terms added and common factors cancelled. Every term on variables $x_1,\ldots,x_n$ can be written as a multichoice
\[ \mulchm{x_1 & \ldots  & x_n \\
        w_1 & \ldots & w_n}
\]
for a unique choice of weights $w_i \in \mathbb N$, not all zero, up to rescaling by a positive rational number. Terms $\pj {-}{x_1 : 0, \ldots, x_n : 0}{t}$ thus correspond precisely to the positive unital linear maps $\RR^n \to \RR$ sending rationals to rationals. }

\subsubsection{Characterization and completeness for $\nu$-free terms}
This section concerns the normalization of terms using free parameters but no $\nu$. Consider a single parameter $p$. If we think of a term $t$ as a syntactic tree, commutativity and discardability can be used to move all occurrences of $\pch p$ to the root of the tree, making it a \emph{tree diagram} of some depth $k$. Let us label the $2^k$ leaves with $t_{a_1\cdots a_k}$, $a_i \in \{0,1\}$. 
As a programming language expression, this corresponds to successive bindings
\begin{align*}
\mlstinline{let $\ a_{1}$=M.get($p$) in ... let $\ a_{k}$=M.get($p$) in}\ t_{a_1\cdots a_k}
\end{align*}
Permutations $\sigma \in S_k$ of the $k$ first levels in the tree act on tree diagrams by permuting the leaves via $t_{a_1\cdots a_k} \mapsto t_{a_{\sigma(1)}\cdots a_{\sigma(k)}}$.
By commutativity~(\ref{eqn:commpchpch}), those permuted diagrams are still equal to $t$, so we can replace $t$ by the average over all permuted diagrams, since rational choice is discardable. The average commutes down to the leaves~(\ref{eqn:commpchrch}), so we obtain a tree diagram with leaves $m_{a_1\cdots a_k} = \frac 1 {k!} \sum_{\sigma} t_{a_{\sigma(1)}\cdots a_{\sigma(k)}}$, where the average is to be read as a rational choice with all weights $1$. This new tree diagram is now by construction invariant under permutation of levels in the tree, in particular $m_{a_1\cdots a_k}$ only depends on the sum $a_1+\dots+a_k$. That is to say, the counts are a sufficient statistic. 

This leads to the following normalization procedure for terms $\pj {p_1 \ldots p_\ell}{x_1 \ldots x_n:0} t$:
Write $C^{p_j}_k(t_0,\ldots,t_k)$ for the permutation invariant tree diagram of $p_j$-choices and depth $k$ with leaves $t_{a_1\cdots a_k} = t_{a_1+\cdots+a_k}$. 
Then we can rewrite $t$ as $C^{p_1}_k(t_0,\ldots,t_k)$ where each $t_i$ is $p_1$-free. Recursively normalize each $t_i$ in the same way, collecting the next parameter. By discardability, we can pick the height of all these tree diagrams to be a single constant $k$, such that the resulting term is a nested structure of tree-diagrams $C^{p_j}_k$. We will use multi-indices $I=(i_1,\ldots,i_\ell)$ to write the whole stratified term as $C_k((t_I))$ where each leaf $t_I$ only contains rational choices. The interpretation of such a term can be given explicitly by Bernstein polynomials
\[ \textstyle \sem{C_k((t_I))}(\vec x)(\vec p) = \sum_I b_{I,k}(\vec p) \cdot \sem{t_I}(\vec x)(\vec p). \]
For example, normalizing $(v\pch p x){\pch p}(y\pch p v)$ gives $(v\pch p(x \rch 1 1 y)){\pch p} ((x\rch 11 y)\pch p v)=C_2(v,x{\rch 11}y,v)$. 

From this we obtain the following completeness result:
\begin{proposition}
\label{prop:bijection-term-function}
There is a bijective correspondence between equivalence classes of terms $\pj {p_1 \ldots p_\ell}{x_1 \ldots x_n:0} t$ and linear unital maps $\phi : \RR^n \to \RR^{\II^\ell}$ such that for every standard basis vector $e_j$ of $\RR^n$, $\phi(e_j)$ is a Bernstein polynomial with nonnegative rational coefficients. 
\end{proposition}
\begin{proof}
We can assume all basis polynomials to have the same degree $k$.
If $\phi(e_j) = \sum_I w_{Ij}b_{I,k}$, then the unitality condition $\phi(1,\ldots,1) = 1$ means $\sum_I \left(\sum_j w_{Ij} \right) b_{I,k} = 1$,
and hence by linear independence and partition of unity, $\sum_j w_{Ij} = 1$ for every $I$. If we thus let $t_I$ be the rational convex combination of the $x_j$ with weights $w_{Ij}$, then $\sem{C_k((t_I))} = \phi$. Again by linear independence, the weights $w_{IJ}$ are uniquely defined by $\phi$.
\end{proof}
Geometric characterizations for the assumption of this theorem exist in \cite{poly-pos-interval, poly-pos-boxes}. For example, a univariate polynomial is a Bernstein polynomial with nonnegative coefficients if and only if it is positive on $(0,1)$. More care is required in the multivariate case.

\subsubsection{Normalization of Beta-Bernoulli}
\label{sec:norm}
For arbitrary terms ${\pj {p_1 \ldots p_\ell} {x_1\colon m_1, \ldots, x_s\colon m_s} t}$, we employ the following normalization procedure. Using conjugacy and the commutativity axioms~(\ref{eqn:commnunu}--\ref{eqn:commnurch}), we can push all uses of $\nu$ towards the leaves of the tree, until we end up with a tree of ratios and free parameter choices only. Next, by conjugacy and discardability, we expand every instance of $\nu_{i,j}$ until they satisfy $i+j = n$ for some fixed, sufficiently large $n$. We then stratify the free parameters into permutation invariant tree diagrams. That is, we find a number $k$ such that $t$ can be written as $C_k((t_I))$ where the leaves $t_I$ consist of $\nu$ and rational choices only. \\

In each $t_I$, commuting all the choices up to the root, we are left with a convex combination of chains of $\nu$'s of the form $\nu_{i_1,j_1} p_{\ell+1}.\ldots\nu_{i_d,j_d} p_{\ell + d}.x_j(p_{\tau(1)},\ldots,p_{\tau(m)})$ for some $\tau : m \to \ell + d$. By discardability, we can assume that there are no unused bound parameters. We consider two chains equal if they are $\alpha$-convertible into each other. Now if $c_1,\ldots,c_m$ is a list of the distinct chains that occur in any of the leaves, we can give the leaves $t_I$ the uniform shape $t_I = \mulchm{c_1 & \ldots & c_m \\ w_{I1} & \ldots & w_{Im}}$ for appropriate weights $w_{Ij} \in \mathbb N$. We will show that this representation is a unique normal form.

\subsubsection{Proof of completeness}\label{sec:completenessproof}
Consider a \emph{chain} $c = \nu_{i_1,j_1} p_{\ell+1}.\ldots\nu_{i_d,j_d} p_{\ell + d}.\,x(p_{\tau(1)},\ldots,p_{\tau(m)})$. Its measure-theoretic interpretation $\sem{c}(p_1,\ldots,p_\ell)$ is a pushforward of a product of $d$ beta distributions, supported on a hyperplane segment that is parameterized by the map $h_\tau : \II^d \to \II^m, h_\tau(p_{\ell+1},\ldots,p_{\ell+d}) = (p_{\tau(1)},\ldots,p_{\tau(m)})$.
Note that the position of the hyperplane may vary with the free parameters. To capture this geometric information, we call $\tau$ the \emph{subspace type} of the chain and $d$ its \emph{dimension}. Because of $\alpha$-invariance of chains, we identify subspace types that differ by a permutation of $\{\ell+1,\ldots,\ell+d\}$. 

\begin{wrapfigure}[11]{r}{5.5cm}\vspace{-3mm}\ \ 
\begin{minipage}{.5\textwidth}
	\begin{tikzpicture}[scale=4]
		\draw[color=black!60!green,fill opacity=0.15,fill=black!60!green] (0,0) rectangle (1,1);

		\draw[very thick,black!40!red] (0,0.8) node[left]{\color{black}$p_2$} -- (1,0.8) node[right] {\footnotesize$(3,2)$};
		\draw[very thick,black!40!red] (0,0.4) node[left] {\color{black}$p_1$}  -- (1,0.4) node[right] {\footnotesize $(3,1)$};
		\draw[very thick,black!20!yellow!20!red] (0.8,0) node[below]{\color{black}$p_2$} -- (0.8,1) node[above] {\footnotesize$(2,3)$};
		\draw[very thick,black!20!yellow!20!red] (0.4,0) node[below]{\color{black}$p_1$} -- (0.4,1) node[above] {\footnotesize$(1,3)$};
		\draw[very thick,black!40!blue] (0,0) node[below] {\footnotesize$(3,3)$} -- (1,1);
		
		\fill[black] (0.4,0.8) circle (0.75pt) node[above left] {\footnotesize $(1,2)$};
		\fill[black] (0.8,0.8) circle (0.75pt) node[above left] {\footnotesize $(2,2)$};
		\fill[black] (0.8,0.4) circle (0.75pt) node[below left] {\footnotesize $(2,1)$};
		\fill[black] (0.4,0.4) circle (0.75pt) node[above left] {\footnotesize $(1,1)$};
		
		\draw[black!60!green] (1.1,0.1) node {\footnotesize $(3,4)$};
	\end{tikzpicture}
\end{minipage}
\end{wrapfigure}
For example, each chain with two free parameters $p_1,p_2$ and a variable $x:2$ gives rise to a parameterized distribution on the unit square.
On the right, we illustrate the ten possible supports that such distributions can have, as subspaces of the square. 
In the graphic we write $(i,j)$ for $\nu p_3.\nu p_4.x(p_i,p_j)$, momentarily omitting the subscripts of $\nu$ because they do not affect the support. For instance, the upper horizontal line corresponds to $\nu{p_3}.x(p_3,p_2)$; the bottom-right dot corresponds to $x(p_2,p_1)$;
the diagonal corresponds to $\nu{p_3}.x(p_3,p_3)$; 
and the entire square corresponds to $\nu p_3.\nu p_4.x(p_3,p_4)$. All told there are four subspaces of dimension~$d=0$, five with~$d=1$, and one with~$d=2$. Notice that the subspaces are all distinct as long as $p_1\neq p_2$.

\begin{proposition}\label{prop:chainli}
If $c_1,\ldots,c_s$ are distinct chains with $i_1+j_1 = \dots = i_d + j_d = n$, then the family of functionals $\{ \sem{c_i}(-)(\vec p) : \RR^{\II^{m_1}} \times \cdots \times \RR^{\II^{m_s}} \to \RR \}_{i = 1,\ldots,s}$ is linearly independent whenever all parameters $p_i$ are distinct.
\end{proposition}
\begin{proof}
Fix $\vec p$. Chains on different variables are clearly independent, so we can restrict ourselves to a single variable $x:m$. We reason measure-theoretically. The interpretation of a chain $c_i$ of subspace type $\tau_i$ is a pushforward measure $h_{i*}(\mu_i)$ where $\mu_i$ is a product of $d$ beta distributions, and $h_i$ is the affine inclusion map $h_i(p_{\ell + 1},\ldots,p_{\ell + d}) = (p_{\tau_i(1)},\ldots,p_{\tau_i(m)})$. Let $\sum a_i h_{i*}(\mu_i) = 0$ as a signed measure. We show by induction over the dimension of the chains that all $a_i$ vanish. Assume that $a_i = 0$ whenever the dimension of $c_i$ is less than $d$, and consider an arbitrary subspace $\tau_j$ of dimension $d$. We can define a signed Borel measure on $\II^d$ by restriction
\[ 
\rho(A) \defeq \sum_i a_i h_{i*}(\mu_i)(h_j(A)) = \sum_i a_i \mu_i(h_i^{-1}(h_j(A)))
\]
as $h_j$ sends Borel sets to Borel sets (e.g.~\cite[\S15A]{kechris}). We claim that $\rho(A) = \sum_{c_i \text{ has type } \tau_j} a_i \mu_i(A)$, as the contributions of chains $c_i$ of different type vanish.
\begin{itemize}
\item If $c_i$ has dimension $< d$, $a_i = 0$ by the inductive hypothesis.
\item If $c_i$ has dimension $> d$, we note that $h_i^{-1}(h_j(A))$ only has at most dimension $d$. It is therefore a nullset for $\mu_i$.
\item If $c_i$ has dimension $d$ but a different type, and all $p_1,\ldots,p_\ell$ are assumed distinct, then the hyperplanes given by $h_i$ and $h_j$ are not identical. Therefore their intersection is at most $(d-1)$-dimensional and $h_i^{-1}(h_j(A))$ is a nullset for $\mu_i$.
\end{itemize}

By assumption, $\rho$ has to be the zero measure, but the $\mu_i$ are linearly independent. Therefore $a_i = 0$ for all $c_i$ with subspace type $\tau_j$. Repeat this for every subspace type of dimension $d$ to conclude overall linear independence.
\end{proof}

\begin{theorem}[Completeness]\label{thm:model-completeness}
If $\pj {\Gamma} {\Delta} {t,t'}$ and $\sem{t} = \sem{t'}$, then $\pj{\Gamma}{\Delta}{t = t'}$.
\end{theorem}
\begin{proof}
	From the normalization procedure, we find numbers $k,n$, a list of distinct chains $c_1,\ldots,c_s$ with $i+j=n$ and weights $(w_{Ij}),(w'_{Ij})$ such that $\pj{\Gamma}{\Delta}{t = C_k((t_I))}$ and $\pj{\Gamma}{\Delta}{t' = C_k((t'_I))}$ where $t_I = \mulchm{c_1 & \ldots & c_s \\ w_{I1} & \ldots & w_{Is}}$ and $t'_I = \mulchm{c_1 & \ldots & c_s \\ w'_{I1} & \ldots & w'_{Is}}$.
The interpretations of these normal forms are given explicitly by 
\[ 
\sem{t}(\vec f)(\vec p) = \sum_{j} \frac{w_{Ij}}{w_I} \cdot b_{I,k}(\vec p) \cdot \sem{c_j}(\vec f)(\vec p) \text{ where } w_I = \sum_j w_{Ij}
\]
and analogously for $t'$. Then $\sem{t} = \sem{t'}$ implies that for all $\vec f$
\[ \sum_j \left(\sum_I \left(\frac{w_{Ij}}{w_I} -\frac{w'_{Ij}}{w'_I} \right) b_{I,k}(\vec p) \right) \sem{c_j}(\vec f)(\vec p) = 0. \]
By Proposition \ref{prop:chainli}, this implies  $\sum_I \left(\frac{w_{Ij}}{w_I} -\frac{w'_{Ij}}{w'_I} \right) b_{I,k}(\vec p) = 0$ for every $j$ and whenever the parameters $p_i$ are distinct.
By continuity of the left hand side, the expression in fact has to vanish for \emph{all} $\vec p$. By linear independence of the Bernstein polynomials, we obtain
$w_{Ij}/w_I = w'_{Ij}/w'_I$ for all $I,j$. Thus, all weights agree up to rescaling and we can conclude $\pj{\Gamma}{\Delta}{t = t'}$.
\end{proof}


\section{Extensionality and syntactical completeness}
\label{sec:extensionality-syntactic-completeness}
In this section we use the model completeness of the previous section to establish some syntactical results about 
the theory of Beta-Bernoulli. Although the model is helpful in informing the proofs, the statements of the results in
this section are purely syntactical. 

The ultimate result of this section is equational syntactical completeness (Cor.~\ref{cor:hpcomplete-bb}), which says that there can be no further equations in the theory without it becoming inconsistent with discrete probability. 
In other words, assuming that the axioms we have included are appropriate, they must be sufficient, regardless of 
any discussion about semantic models or intended meaning. This kind of result is sometimes called `Post completeness' after Post proved a similar result for propositional logic. 

The key steps towards this result are two extensionality results. These are related to the programming language idea 
of `contextual equivalence'.
Recall that in a programming language we often define a basic notion of 
equivalence on closed ground terms: these are programs with no free variables that return (say) booleans. 
This notion 
is often defined by some operational consideration using some notions of observation. From this we define contextual equivalence by saying that $t\approx u$ if, for all closed ground contexts $\mathcal C$, 
$\mathcal C[t]=\mathcal C[u]$. 

Contextual equivalence has a canonical appearance, but an axiomatic theory of equality, such as the one in this paper, is more compositional and easier to work with. Our notion of equality induces in particular a basic notion of 
equivalence on closed ground terms. Our extensionality results say that, assuming one is content with this basic notion of equivalence, the equations that we axiomatize coincide with contextual equivalence. 

\subsection{Extensionality}
\begin{proposition}[Extensionality for closed terms]
\label{prop:ext-closed} Suppose 
$\pj {\Gamma,q}\Delta {t}$ and $\pj {\Gamma,q}\Delta {u}$. 
If $\pj {\Gamma}\Delta {\nu_{i,j}q.t=\nu_{i,j}q.u}$ for all $i,j$,
then also $\pj {\Gamma}\Delta {t=u}$.
\end{proposition}
\begin{proof}
We show the contrapositive. By the model completeness theorem (Thm.~\ref{thm:model-completeness}), we can reason in the model rather than syntactically. 
So we consider $t$ and $u$ such that $\sem t\neq \sem u$ as functions
$\RR^{\II^{m_1}}\times \RR^{\II^{m_k}}\to \RR^{\II^{l+1}}$, and show that there are $i,j$ such that 
$\sem {\nu_{i,j}q.t}\neq\sem{\nu_{i,j}q.u}$.
By assumption there are $\vec f$ and $\vec p,q$ such that $\sem t(\vec f)(\vec p,q)\neq \sem u(\vec f)(\vec p,q)$
as real numbers.

Now we use the following general reasoning: 
For any real $q\in \II$
we can pick monotone sequences $i_1<\dots<i_n<\dots$ and 
$j_1<\dots<j_n<\dots$ of natural numbers so that $\frac{i_n}{i_n+j_n}\to q$ as $n\to \infty$.
Moreover, for any continuous $h:\II\to \RR$, the integral $\int h \ \dd\dbeta {i_n}{j_n}$ converges to $h(q)$ as $n\to \infty$: one way to see this is to notice that the variance
of $\dbeta {i_n}{j_n}$ vanishes as $n\to \infty$, so by Chebyshev's inequality, $\lim_n \dbeta {i_n}{j_n}$ is a Dirac distribution at $q$. Thus, $\int\big(\sem t(\vec f)(\vec p,r)-\sem u(\vec f)(\vec p,r)\big)\ \dbeta{i_n}{j_n}(\dd r)$ is non-zero
as $n\to \infty$. By continuity, for some $n$, 
$\int \sem t(\vec f)(\vec p,r)\ \dbeta{i_n}{j_n}(\dd r)\neq 
\int
\sem u(\vec f)(\vec p,r)\ \dbeta{i_n}{j_n}(\dd r)$. 
	So, $\sem {\nu_{i_n,j_n}q.t}\neq\sem{\nu_{i_n,j_n}q.u}$.
\end{proof}

\begin{proposition}[Extensionality for ground terms]
\label{prop:ext-ground}
\emph{In brief: }If ${t[^{v_1 \dots v_k}\!/\!_{x_1\dots x_k}]
= 
u[^{v_1 \dots v_k}\!/\!_{x_1\dots x_k}]}$
for all suitable ground $v_1\dots v_k$, then $ t =u$. 

\emph{In detail: }Consider $t$ and $u$ with
$\pj {-}{x_1\colon m_1\dots x_k\colon m_k} {t,u}$. 
Suppose that whenever $v_1\dots v_k$ are terms with
$(\pj{p_1\dots p_{m_1}}{y,z:0}{v_1})$, \dots,\ 
$(\pj{p_1\dots p_{m_k}}{y,z:0}{v_k})$,
then we have 
$\pj{-}{y,z:0}{
t[^{v_1 \dots v_k}\!/\!_{x_1\dots x_k}]
= 
u[^{v_1 \dots v_k}\!/\!_{x_1\dots x_k}]}\text.
$
Then we also have
$
{\pj {-}{x_1\colon m_1\dots x_k\colon m_k} {t=u}\text.}
$
\end{proposition}
\begin{proof}
Again, we show the contrapositive. 
Let $\Delta = (x_1\colon m_1\dots x_k\colon m_k)$.
Suppose we have $t$ and $u$ such that $\neg(\pj{-}{\Delta}{t=u})$. Then by the model completeness theorem (Thm.~\ref{thm:model-completeness}),
we have $\sem t\neq \sem u$ as linear functions $\RR^{\II^{m_1}}\times\dots\times \RR^{\II^{m_k}}\to \RR$. 
Since the functions are linear, there is an index $i\leq k$ and a continuous function $f:{\II^{m_i}}\to \RR$ with 
$\sem t(0\dots 0,f,0\dots 0)
\neq 
\sem u(0\dots 0,f,0\dots 0)$. 
By the Stone-Weierstrass theorem, every such $f$ is a limit of polynomials,
and so since $\sem t$ and $\sem u$ are continuous and linear, 
there has to be a Bernstein basis polynomial $b_{I,k}: {\RR^{m_i}}\to \RR$ that already distinguishes them. This function is definable, i.e. there is a a term $\pj{p_1,\ldots,p_{m_i}}{y,z:0} w$ with $\sem{w}(1,0) = b_{I,k}$. Define terms $v_j = w$ for $i=j$ and $v_j=z$ for $i\neq j$. Then
\[
\sem {t[^{v_1\dots v_k}\!/\!_{x_1\dots x_k}]}(1,0) = 
\sem t(0, {\dots},b_{I,k},{\dots}, 0)
\neq 
\sem u(0,{\dots}, b_{I,k},{\dots}, 0)
=
\sem {u[^{v_1\dots v_k}\!/\!_{x_1\dots x_k}]}(1,0).
\]
The required
${\neg
\big(\pj{-}{y,z:0}{
t[^{v_1 \dots v_k}\!/\!_{x_1\dots x_k}]
= 
u[^{v_1 \dots v_k}\!/\!_{x_1\dots x_k}]}\big)}$
follows from the above disequality
because of the model soundness property (Props.~\ref{prop:model-soundness1}
and \ref{prop:model-soundness2}).
\end{proof}
From the programming perspective, a term $\pj-{y,z:0}{t_0}$ corresponds to a closed program of type \lstinline|bool|,
for it has two possible continuations, $y$ and $z$, depending on whether the outcome is \lstinline|true| or \lstinline|false|. 
From this perspective, Proposition~\ref{prop:ext-ground} says that for closed $t,u$, 
if $\mathcal C[t]=\mathcal C[u]$ for all boolean contexts $\mathcal C$, 
then $t=u$. 

\subsection{Relative syntactical completeness}
\label{sec:rel-syn-comp}
\begin{proposition}[Neumann, {\cite{neumann}}]
\label{prop:hpcomplete-conv}
If $t,u$ are terms in the theory of rational convexity (Def.~\ref{def:rational-convex}), then either $t=u$ is derivable or 
it implies $x\rch ij y=x\rch {i'}{j'} y$ for all nonzero $i,i',j,j'$. 
\end{proposition}

\begin{corollary}
\label{cor:hpcomplete-bb}
The theory of Beta-Bernoulli is syntactically complete relative to the theory of rational convexity, in the following sense. 
For all terms $t$ and $u$, either $t=u$ is derivable, or it implies 
$x\rch ij y=x\rch {i'}{j'} y$ for all nonzero $i,i',j,j'$. 
\end{corollary}
This is proved by combining Propositions \ref{prop:ext-closed}, \ref{prop:ext-ground} and \ref{prop:hpcomplete-conv}. As an example for extensionality and completeness, consider the equation $\nu_{1,1}p.x(p,p) = \nu_{1,1}p.(\nu_{1,1}q.x(p,q))$. It is not derivable, as can be witnessed by the substitution $x(p,q) = (y\pch q z) \pch p z$. Normalizing yields $y \rch 1 2 z = y \rch 1 3 z$ which is incompatible with discrete probability (see \S\ref{app:derivation}). 
In programming syntax, the candidate equation is written
\\\\\begin{minipage}{.47\textwidth}
\begin{lstlisting}
LHS = let p = M.new(1,1) in (p,p)
\end{lstlisting}
\end{minipage}\,
\begin{minipage}{.53\textwidth}
\begin{lstlisting}
RHS = (M.new(1,1) , M.new(1,1))
\end{lstlisting}
\end{minipage}
and the distinguishing context is {\lstinline|C[-] = let (p,q)=(-) in if M.get(p) then M.get(q) else false|}. 
That is to say, the closed ground programs \lstinline|C[LHS]| and \lstinline|C[RHS]| necessarily have different observable statistics: this follows from the axioms. 

\subsection{Remark about stateful implementations}
In the introduction we recalled the idea of using P\'olya's urn to implement a Beta-Bernoulli process
using local (hidden) state. 

Our equational presentation gives a recipe for understanding the correctness of the stateful 
implementation. First, one would give an operational semantics,
and then a basic notion of observational equivalence on closed ground terms in terms of the finite probabilities associated with reaching certain ground values. From this, an operational notion of contextual equivalence 
can be defined~(e.g.~\cite[\S 6]{BizjakBirkedal}, \cite{sv, wgcc}).
Then, one would show that the axioms of our theory hold up-to contextual equivalence. 
Finally one can deduce from the syntactical completeness result that the equations satisfied by this stateful implementation must be exactly the equations satisfied by the semantic model. 

In fact, in this argument, 
it is not necessary to check that axioms~(C1) and~(D2) hold in the operationally defined contextual equivalence, because 
the axiomatized equality on closed ground terms is independent of these axioms. To see this, notice that our normalization procedure (\S\ref{sec:norm}) 
doesn't use~(C1) or~(D2) when the terms are closed and ground, since then we can take $n=k=0$. 
This is helpful because the remaining axioms are fairly straightforward, e.g.~\eqref{eqn:conjugacy} is the essence of the urn scheme and (D1) is garbage collection. 


\section{Conclusion}
\label{sec:conc}
Exchangeable random processes are central to many Bayesian models.
The general message of this paper is that the analysis of exchangeable random processes, based on basic concepts from programming language theory, depends on three crucial ingredients: commutativity, discardability, and abstract types. 
We have illustrated this message by showing that just adding the conjugacy law to these ingredients leads to a complete equational theory for the Beta-Bernoulli process (Thm.~\ref{thm:model-completeness}). Moreover, we have shown that this equational theory has a canonical syntactic and axiomatic status, regardless of the measure theoretic foundation (Cor.~\ref{cor:hpcomplete-bb}). Our results in this paper open up the following avenues of research. 
\begin{description}
\item[Study of nonparametric Bayesian models:] 
We contend that abstract types, commutativity and discardability are fundamental tools for studying nonparametric Bayesian models, especially hierarchical ones. 
For example, the Chinese Restaurant Franchise~\cite{Teh-HDP} can be implemented as a module with three abstract types, 
\lstinline|f| (franchise), \lstinline|r| (restaurant), \lstinline|t| (table), and functions
\lstinline|newFranchise:()->f|, \lstinline|newRestaurant:f->r|, \lstinline|getTable:r->t|, \lstinline|sameDish:t*t->bool|.
Its various exchangeability properties correspond to commutativity/discardability in the presence of type abstraction. (For other examples, see \cite{XRPDA-PPS2017}.)
\item[First steps in synthetic probability theory:]
As is well known, the theory of rational 
convex sets corresponds to the monad~$D$ of rational discrete probability distributions. 
Commutativity of the theory amounts to commutativity of the monad~$D$~\cite{linton-comm,kock-comm}. 

As any parameterized algebraic theory, the theory of Beta-Bernoulli (\S\ref{sec:presentation-beta-bernoulli}) 
can be understood as a monad~$P$ on the functor category $[\mathbf{FinSet},\mathbf{Set}]$,
with the property that to give a natural transformation~$\mathbf{FinSet}(\ell,-)\to P(\coprod_{j=1}^k\mathbf{FinSet}(m_k,-))$ is to give a term $(\pj{p_1\dots p_\ell}{x_1\colon m_1\dots x_k\colon m_k}t)$, and monadic bind is substitution
(\cite[Cor.~1]{s-pred-logic}, \cite[\S VIIA]{s-instances}). This can be thought of as an intuitionistic set theory with an interesting notion of probability. 
As such this is a `commutative effectus'~\cite{jacobs-commutative-effectus}, a synthetic probability theory (see also~\cite{kock:commutative-monads-as-a-theory-of-distributions}).
 Like~$D$, the
global elements $1\to P(2)$ are the rationals in $[0,1]$ (by Prop.~\ref{prop:bijection-term-function}) but unlike~$D$, the global elements
$1\to P(P(2))$ include the beta distribution.

\item[Practical ideas for nonparametric Bayesian models in probabilistic programming:]
A more practical motivation for our work is to inform the design of module systems for probabilistic programming languages. 
For example, Anglican, Church, Hansei and Venture already support nonparametric Bayesian primitives~%
\cite{ks-prob-first-class-store,wu-church,msp-venture}. 
We contend that abstract types are a crucial concept from the perspective of exchangeability.
\end{description}
\medskip\textbf{Acknowledgements.} It has been helpful to discuss this work with many people, including Ohad Kammar, Gordon Plotkin, Alex Simpson, and Marcin Szymczak. This work is partly supported by EPSRC grants EP/N509711/1 and EP/N007387/1, a Royal Society University Research Fellowship, and an Institute for Information \& communications Technology Promotion (IITP) grant funded by the Korea government (MSIT) (No.2015-0-00565, Development of Vulnerability Discovery Technologies for IoT Software Security).

\bibliography{refs}

\newpage\appendix
\section{Example derivations}
\label{app:derivation}

In this appendix, we derive equations mentioned in the main text of the paper. The first equation, as found in \S\ref{sec:axioms-beta-bernoulli}, is
\[ x \rch 1 1 y = ((x \rch 1 1 y) \pch p x) \pch p (y \pch p (x \rch 1 1 y)) \]
which is called von Neumann's trick. Here is the derivation of this equation:
\begin{align*}
	((x \rch 1 1 y) \pch p x) \pch p (y \pch p (x \rch 1 1 y)) 
	& = ((x \rch 1 1 y) \pch p (x \rch 1 1 x)) \pch p ((y \rch 1 1 y) \pch p (x \rch 1 1 y))
	\\ 
	& = ((x \pch p x) \rch 1 1 (y \pch p x)) \pch p ((y \pch p x) \rch 1 1 (y \pch p y))
	\\ 
	& = ((y \pch p x) \rch 1 1 (x \pch p x)) \pch p ((y \pch p x) \rch 1 1 (y \pch p y))
	\\
	& = ((y \pch p x) \pch p (y \pch p x)) \rch 1 1 ((x \pch p x) \pch p (y \pch p y)) 
	\\
	& = (y \pch p x) \rch 1 1 (x \pch p y) 
	\\
	& = (y \rch 1 1 x) \pch p (x \rch 1 1 y) 
	\\
	& = (x \rch 1 1 y) \pch p (x \rch 1 1 y) 
	\\
	& = (x \rch 1 1 y).
\end{align*}
Using normal forms, we can in fact easily see that $x \rch 1 1 y$ is the only normalized term $(\pj{-}{x,y} t)$ that satisfies $t = (t \pch p x) \pch p (y \pch p t)$. Making $(t \pch p x) \pch p (y \pch p t)$ permutation-invariant means computing the average
\begin{align*}
(t \pch p x) \pch p (y \pch p t)&=
((t \pch p x) \pch p (y \pch p t)) \rch 1 1 ((t \pch p y) \pch p (x \pch p t)) 
\\&=  ((t\rch 11 t) \pch p (x \rch 1 1 y)) \pch p ((x \rch 1 1 y) \pch p (t\rch 11 t)) 
\\&=  (t \pch p (x \rch 1 1 y)) \pch p ((x \rch 1 1 y) \pch p t) 
\ \ = C_2(t,x \rch 1 1 y, t) 
\end{align*}
whereas by discardability, the permutation-invariant form of $t$ is simply $C_2(t,t,t)$.
As these are normal forms, we can read off that $t=x \rch 1 1 y$ is the only fixed point.\\

The following derivations show the normalization procedure being applied to the programs from \S\ref{sec:rel-syn-comp}:
\begin{figure}[!htb]
\begin{minipage}{.5\textwidth}
\begin{align*}
\quad & \nu_{1,1}p. ((y \pch p z) \pch p z) \\
& = (\nu_{2,1}p. (y \pch p z)) \rch 1 1 (\nu_{1,2}p. z) \\
& = (y \rch 2 1 z) \rch 1 1 z \\
& = (y \rch 2 1 z) \rch 3 3 (y \rch 0 3 z) \\ 
& = (y \rch 2 0 y) \rch 2 4 (z \rch 1 3 z) \\
& = y \rch 2 4 z \\
& = y \rch 1 2 z. 
\end{align*}
\end{minipage}\,
\begin{minipage}{.5\textwidth}
\begin{align*}
\quad &\nu_{1,1}p. \nu_{1,1}q. ((y \pch q z) \pch p z) \\
& = \nu_{1,1}p. ((\nu_{1,1}q. (y \pch q z)) \pch p (\nu_{1,1}q.z)) \\
& = \nu_{1,1}p. ((y \rch 1 1 z) \pch p z) \\
& = (y \rch 1 1 z) \rch 1 1 z \\
& = (y \rch 1 1 z) \rch 2 2 (y \rch 0 2 z) \\
& = (y \rch 1 0 y) \rch 1 3 (z \rch 1 2 z) \\
& = y \rch 1 3 z.
\end{align*}
\end{minipage}
\end{figure}

\pagebreak

\end{document}